\documentclass{llncs}
\usepackage{amsmath}
\usepackage{amssymb}
\usepackage{mathrsfs}
\usepackage{graphicx}
\usepackage{color}
\usepackage[all]{xy}

\newcommand{\shefour}[4]{
\resizebox{!}{.5cm}{
\ensuremath{
\begin{array}{c|c}
1 & #1 \\
\hline
2 & #2 \\
\hline
3 & #3 \\
\hline
4 & #4 
\end{array}
}
}
}

\newcommand{\Ret}[0]{\ensuremath{\textsc{Ret}}}

\renewcommand{\phi}{\varphi}

\newcommand{\ignore}[1]{}

\title{QCSP on partially reflexive cycles - the wavy line of tractability}
\author{Florent Madelaine\inst{1} \and Barnaby Martin\inst{2}\thanks{The author is supported by ANR Blanc International ALCOCLAN.}}
\institute{Clermont Universit{\'e}, Universit{\'e} d'Auvergne, LIMOS, BP 10448, F-63000 Clermont-Ferrand, France.\and CNRS / LIX UMR 7161, \'Ecole Polytechnique, Palaiseau, France.}

\begin{document}
\maketitle

\begin{abstract}
We study the (non-uniform) quantified constraint satisfaction problem QCSP$(\mathcal{H})$ as $\mathcal{H}$ ranges over partially reflexive cycles. We obtain a complexity-theoretic dichotomy: QCSP$(\mathcal{H})$ is either in NL or is NP-hard. The separating conditions are somewhat esoteric hence the epithet ``wavy line of tractability'' (see Figure~\ref{fig:wavyline} at end). 
\end{abstract}

\section{Introduction}

The \emph{quantified constraint satisfaction problem} QCSP$(\mathcal{B})$, for a fixed \emph{template} (structure) $\mathcal{B}$, is a popular generalisation of the \emph{constraint satisfaction problem} CSP$(\mathcal{B})$. In the latter, one asks if a primitive positive sentence (the existential quantification of a conjunction of atoms) $\Phi$ is true on $\mathcal{B}$, while in the former this sentence may be positive Horn (where universal quantification is also permitted). Much of the theoretical research into CSPs is in respect of a large complexity classification project -- it is conjectured that CSP$(\mathcal{B})$ is always either in P or NP-complete \cite{FederVardi}. This \emph{dichotomy} conjecture remains unsettled, although dichotomy is now known on substantial classes (e.g. structures of size $\leq 3$ \cite{Schaefer,Bulatov} and smooth digraphs \cite{HellNesetril,barto:1782}). Various methods, combinatorial (graph-theoretic), logical and universal-algebraic have been brought to bear on this classification project, with many remarkable consequences. A conjectured delineation for the dichotomy was given in the algebraic language in \cite{JBK}.

Complexity classifications for QCSPs appear to be harder than for CSPs. Indeed, a classification for QCSPs will give a fortiori a classification for CSPs (if $\mathcal{B} \uplus \mathcal{K}_1$ is the disjoint union of $\mathcal{B}$ with an isolated element, then QCSP$(\mathcal{B} \uplus \mathcal{K}_1)$ and CSP$(\mathcal{B})$ are polynomially equivalent). Just as CSP$(\mathcal{B})$ is always in NP, so QCSP$(\mathcal{B})$ is always in Pspace. However, no overarching polychotomy has been conjectured for the complexities of QCSP$(\mathcal{B})$, as $\mathcal{B}$ ranges over finite structures, but the only known complexities are P, NP-complete and Pspace-complete (see \cite{BBCJK,CiE2006} for some trichotomies). It seems plausible that these complexities are the only ones that can be so obtained (for more in this see \cite{Meditations}). 

In this paper we study the complexity of QCSP$(\mathcal{H})$, where $\mathcal{H}$ is a partially reflexive cycle. In this respect, our paper is a companion to the similar classification for partially reflexive forests in \cite{QCSPforests}. We derive a classification between those cases that are in NL and those that are NP-hard. For some of the NP-hard cases we are able to demonstrate Pspace-completeness. The dichotomy, as depicted in Figure~\ref{fig:wavyline} at the end, is quite esoteric and deviates somewhat from similar classifications (\mbox{e.g.} for retraction as given in \cite{pseudoforests}). To our minds, this makes it interesting in its own right. Some of our hardness proofs come from judicious amendments to the techniques used in \cite{QCSPforests}. Several others use different elaborate encodings of retraction problems, known to be hard from \cite{pseudoforests}. All but one of our NL-membership results follow from a majority polymorphism in an equivalent template (indeed -- the so-called \emph{Q-core} of \cite{CP2012}), as they did in \cite{pseudoforests}. However, $\mathcal{C}_{0111}$ is special. It has no QCSP-equivalent that admits a majority (indeed, it omits majority and is a Q-core), so we have to give a specialised algorithm, based on ideas from \cite{LICS2011}, to demonstrate that QCSP$(\mathcal{C}_{0111})$ is in L. Indeed, and in light of the observations in \cite{CP2012}, this is the principal news from the partially reflexive cycles classification that removes it from being simply a sequel to partially reflexive forests: for a partially reflexive forest $\mathcal{H}$, either the Q-core of $\mathcal{H}$ admits a majority polymorphism and QCSP$(\mathcal{H})$ is in NL, or QCSP$(\mathcal{C})$ is NP-hard. The same can not be said for partially reflexive cycles, due to the odd case of $\mathcal{C}_{0111}$. 

This paper is organised as follows. After the preliminaries, we address small cycles in Section~\ref{sec:small}. Then we deal with reflexive cycles, cycles whose loops induce a path and cycles with disconnected loops in Sections~\ref{sec:reflexive}, \ref{sec:path} and  \ref{sec:disconnected}, respectively. Finally we give our classification in Section~\ref{sec:class} and our conclusions in Section~\ref{sec:conclusion}. For reasons of space, many proofs are  deferred to the appendix.

\section{Definitions and preliminaries}

Let $[n]:=\{1,\ldots,n\}$. A graph $\mathcal{G}$ has vertex set $G$, of cardinality $|G|$, and edge set $E(\mathcal{G})$.
For a sequence $\alpha \in \{0,1\}^*$, of length $|\alpha|$, let $\mathcal{P}_\alpha$ be the undirected path on $|\alpha|$ vertices such that the $i$th vertex has a loop iff the $i$th entry of $\alpha$ is $1$ (we may say that the path $\mathcal{P}$ is \emph{of the form} $\alpha$). We will usually envisage the domain of a path with $n$ vertices to be $[n]$, where the vertices appear in the natural order. Similarly, for $\alpha \in \{0,1\}^*$, let $\mathcal{C}_{\alpha}$ be the $|\alpha|$ cycle with domain $[n]$ and edge set $\{ (i,j) : |j-i|=1 \bmod n\} \cup \{ (i,i) : \alpha[i]=1 \}$ (note $|n-1|=|1-n|=1 \bmod n$). If $\alpha$ and $\alpha'$ are sequences in $\{0,1\}^n$ such that $\alpha[i]=\alpha'[i+1 \bmod n]$ then $\mathcal{C}_{\alpha}$ and $\mathcal{C}_{\alpha'}$ are isomorphic.

A partially reflexive cycle is one that may include some self-loops. For such an $m$-cycle $\mathcal{C}$, whose vertices we will imagine to be $v_1,\ldots,v_m$ in their natural $\bmod\ m$ adjacencies, let $[v_i \Rightarrow v_j]$ be shorthand for a conjunction specifying a path, whichever is the fastest way $\bmod\ m$, from $v_i$ to $v_j$. For example, if $m=5$, then 1.) $[v_1 \Rightarrow v_3]$ is $E(v_1,v_2) \wedge E(v_2,v_3)$, 2.) $[v_3 \Rightarrow v_1]$ is $E(v_3,v_2) \wedge E(v_2,v_1)$, and 3.) $[v_4 \Rightarrow v_1]$ is $E(v_4,v_5) \wedge E(v_5,v_1)$. We ask the reader to endure the following relaxation of this notation; $[v_i,v_{i+1} \Rightarrow v_j]$ indicates an edge from $v_i$ to $v_{i+1}$ then a path to $v_j$ (which may not be the same as $[v_i \Rightarrow v_j]$ as the latter may take the other path around the cycle). Finally, let $\mathrm{Ref}(v_i,\ldots,v_j)$ indicate $E(v_i,v_i) \wedge \ldots \wedge E(v_j,v_j)$, whichever is the quickest way around the cycle $\bmod\ m$. All graphs in this paper are undirected, so edge statements of the form $E(x,y)$ should be read as asserting $E(x,y) \wedge E(y,x)$.

The problems CSP$(\mathcal{T})$ and QCSP$(\mathcal{T})$ each take as input a sentence $\Phi$, and ask whether this sentence is true on $\mathcal{T}$. For the former, the sentence involves the existential quantification of a conjunction of atoms -- \emph{primitive positive} logic. For the latter, the sentence involves the arbitrary quantification of a conjunction of atoms -- \emph{positive Horn} logic.  By convention equalities are permitted in both of these, but these may be propagated out by substitution in all but trivial degenerate cases. 
The \emph{retraction problem} $\Ret(\mathcal{B})$ takes as input some $\mathcal{G}$, with $\mathcal{H}$ an induced substructure of $\mathcal{G}$, and asks whether there is a homomorphism $h:\mathcal{G} \rightarrow \mathcal{H}$ such that $h$ is the identity on $\mathcal{H}$. It is important that the copy of $\mathcal{H}$ is specified in $\mathcal{G}$; it can be that $\mathcal{H}$ appears twice as an induced substructure and there is a retraction from one of these instances but not to the other. The problem Ret$(\mathcal{H})$ is easily seen to be logspace equivalent with the problem CSP$(\mathcal{H}^\mathrm{c})$, where $\mathcal{H}^\mathrm{c}$ is $\mathcal{H}$ expanded with all constants (one identifies all elements assigned to the same constant and enforces the structure $\mathcal{H}$ on those constants).

The \emph{direct product} $\mathcal{G} \times \mathcal{H}$ of two graphs $\mathcal{G}$ and $\mathcal{H}$ has vertex set $\{(x,y):x \in G, y \in H\}$ and edge set $\{((x,u),(y,v)):x,y \in G, u,v \in H, (x,y) \in E(\mathcal{G}), (u,v) \in E(\mathcal{H})\}$. Direct products are (up to isomorphism) associative and commutative. The $k$th power $\mathcal{G}^k$ of a graph $\mathcal{G}$ is $\mathcal{G} \times \ldots \times \mathcal{G}$ ($k$ times). A homomorphism from a graph $\mathcal{G}$ to a graph $\mathcal{H}$ is a function $h:G\rightarrow H$ such that, if $(x,y) \in E(\mathcal{G})$, then $(h(x),h(y)) \in E(\mathcal{G})$. A \emph{$k$-ary polymorphism} of a graph is a homomorphism from $\mathcal{G}^k$ to $\mathcal{G}$. A ternary function $f:G^3 \rightarrow G$ is designated a \emph{majority} operation if $f(x,x,y)=f(x,y,x)=f(y,x,x)=x$, for all $x,y \in G$.

A positive Horn sentence $\Phi$ in the language of graphs induces naturally a graph $\mathcal{G}_\Phi$ whose vertices are the variables of $\Phi$ and whose edges are the atoms of $\Phi$. In the case of primitive positive $\Phi$ one would call $\mathcal{G}_\Phi$ the \emph{canonical database} and $\Phi$ its \emph{canonical query}. With positive Horn $\Phi$ there is additional extra structure and one may talk of a vertex-variable as being existential/ universal and as coming before/ after (earlier/ later), in line with the quantifier \textbf{block} and its order in $\Phi$. Variables in the same quantifier block will not need their orders considered (there is commutativity within a block anyway). A typical reduction from a retraction problem Ret$(\mathcal{C})$, where $|C|=m$, builds a positive Horn $\Phi$ sentence involving variables $v_1,\ldots,v_m$ where we want $\mathcal{G}_\Phi$ restricted to $\{v_1,\ldots,v_m\}$ (itself a copy of $\mathcal{C}$) to map automorphically to $\mathcal{C}$. Typically, we can force this with some evaluation of the variables (some of which might be universally quantified). The other valuations are \emph{degenerate} and we must ensure at least that they map $\mathcal{G}_\Phi$ restricted to $\{v_1,\ldots,v_m\}$ homomorphically to $\mathcal{C}$.

\section{Small cycles}
\label{sec:small}

The classification for QCSP for cycles of length $\leq 4$ is slightly esoteric, although it does match the analogous classification for Retraction (the former is a dichotomy between NL and Pspace-complete; the latter is a dichotomy between P and NP-complete). The following has appeared in \cite{CountingQuantifiers}, where it was given as an application of counting quantifiers in QCSPs. We review it now because we will need to generalise it later in Section~\ref{sec:reflexive}.
\begin{proposition}[\cite{CountingQuantifiers}]
\label{prop:C1111}
QCSP$(\mathcal{C}_{1111})$ is Pspace-complete.
\end{proposition}
\begin{proof}
We will reduce from the problem QCSP$(\mathcal{K}_4)$ (known to be Pspace-complete from, e.g., \cite{BBCJK}). We will borrow heavily from the reduction of CSP$(\mathcal{K}_4)$ to Ret$(\mathcal{C}_{1111})$ in \cite{FederHell98}. We introduce the following shorthands ($x',x''$ must appear nowhere else in $\phi$, which may contain more free variables than just $x$).
\[
\begin{array}{l}
\exists^{\geq 1} x \ \phi(x) := \exists x \ \phi(x) \\
\exists^{\geq 2} x \ \phi(x) := \forall x' \exists x \ E(x',x) \wedge \phi(x) \\ 
\exists^{\geq 3} x \ \phi(x) := \forall x'' \forall x' \exists x \ E(x'',x) \wedge E(x',x) \wedge \phi(x) \\ 
\end{array}
\]
On $\mathcal{C}_{1111}$, it is easy to verify that, for each $i \in [4]$, $\exists^{\geq i} x \ \phi(x)$ holds iff there exist at least $i$ elements $x$ satisfying $\phi$. Thus our borrowing the notation of counting quantifiers is justified.

We now reduce an instance $\Phi$ of QCSP$(\mathcal{K}_4)$ to an instance $\Psi$ of QCSP$(\mathcal{C}_{1111})$.

We begin with a cycle $\mathcal{C}_{1111}$ on vertices $1$, $2$, $3$ and $4$, which we realise through their canonical query (without quantification) as $\theta(v_1,v_2,v_3,v_4):=$ $E(v_1,v_2) \wedge E(v_2,v_3) \wedge E(v_3,v_4) \wedge E(v_4,v_1)$ (recall that the canonical query is in fact the reflexive closure of this, but this will not be important in this case or many future cases -- when it is important it will be stated explicitly).
If $\Phi$ contains an atom $E(x,y)$, then this gives rise to a series of atoms in $\Psi$ as dictated by the gadget in Figure~\ref{fig:surhomcref4} (for each atom we add many new vertex-variables, corresponding to the vertices in the gadget that are not $x,y,1,2,3,4$).
\begin{figure}
\begin{center}
\resizebox{!}{2.5cm}{\includegraphics{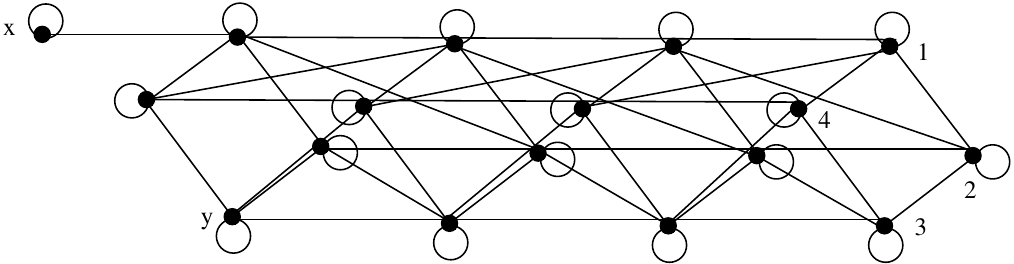}}
\end{center}
\caption{Edge gadget in reduction from QCSP$(\mathcal{K}_4)$ to  QCSP$(\mathcal{C}_{1111})$}
\label{fig:surhomcref4}
\end{figure}
These atoms can be seen to join up with the atoms of $\theta$ as in the right end of the figure. Build $\Psi'''$ from $\Phi$ by this process and conjunction with $\theta$. Then make $\Psi''$ from $\Psi'''$ by existentially quantifying all of the variables other than those associated to atoms of $\Phi$ ($x$, $y$ in the figure) and $v_1,v_2,v_3,v_4$ ($1,\ldots,4$ in the figure). Now, we build $\Psi'$ from $\Psi''$, by copying the quantifier order of $\Phi$ on the outside of the existential quantifiers that we already have. Thus $\Psi'(v_1,v_2,v_3,v_4)$ is a positive Horn formula with precisely four free variables.

It is not hard to see that when $v_1,v_2,v_3,v_4$ are evaluated as (an automorphism of) $1,2,3,4$, then we have a faithful simulation of QCSP$(\mathcal{K}_4)$. This is because $x$ and $y$, as in the gadget drawn, may evaluate precisely to distinct vertices on $\mathcal{C}_{1111}$. Finally, we build $\Psi:= \exists v_1 \exists^{\geq 2} v_2 \exists^{\geq 3} v_3 \exists^{\geq 2} v_4 \ \Psi'(v_1,v_2,v_3,v_4)$.
It is not hard to see that $\Psi$ forces on some evaluation of $v_1,v_2,v_3,v_4$ that these map isomorphically to $1,2,3,4$ in $\mathcal{C}_{1111}$. Further, a rudimentary case analysis shows us that when they do not, we can still evaluate the remainder of $\Psi'$, if we could have done when they did. In fact, if $v_1,v_2,v_3,v_4$ are not mapped isomorphically (but still homomorphically, of course) to $1,2,3,4$ then we can extend each of the edge gadgets to homomorphism under \textbf{all maps} of vertices $x$ and $y$ to $1,2,3,4$ (not just ones in which $x$ and $y$ are evaluated differently). 
\end{proof}
\begin{proposition}
\label{prop:C0111}
QCSP$(\mathcal{C}_{0111})$ is in L.
\end{proposition}
\begin{proof}
Recall $\mathcal{C}_{0111}$ has vertices $\{1,2,3,4\}$ in cyclic order and $1$ is the only non-loop. 
Consider an input $\Phi$ for QCSP$(\mathcal{C}_{0111})$, \mbox{w.l.o.g.} without any equalities, and its evaluation on $\mathcal{C}_{0111}$ as a game on $\Phi$ between \emph{Prover}, playing (evaluating on $\mathcal{C}_{0111}$) existential variables, and \emph{Adversary}, playing universal variables. Adversary never gains by playing $3$, as any existential edge witness to anything from $\{4,1,2\}$ is already an edge-witness to $3$. That is, if $E(x,3)$ then already each of $E(x,4), E(x,1)$ and $E(x,2)$. Similarly, Prover never gains by playing $1$. Thus, $\Phi$ is true on  $\mathcal{C}_{0111}$ iff it is true with all universal variables relativised to  $\{4,1,2\}$ and all existential variables relativised to  $\{2,3,4\}$. (This intuition is formalised in the notion of $U$-$X$-surjective hyper-endomorphism in \cite{LICS2011}. What we are saying is that $\shefour{13}{2}{3}{4}$ is a surjective hyper-endomorphism of $\mathcal{C}_{0111}$.) Henceforth, we will make this assumption of relativisation in our inputs.

Given an input $\Phi$ we will describe a procedure to establish whether it is true on $\mathcal{C}_{0111}$ based around a list of forbidden subinstances.
\begin{itemize}
\item[$(i.)$] An edge $E(x,y)$ in $\mathcal{G}_\Phi$ with the later of $x$ and $y$ being universal.
\item[$(ii.)$] A $3$-star $E(x_1,y), E(x_2,y), E(x_3,y)$ where $x_1,x_2,x_3$ are universal variables coming before $y$ existential.
\item[$(iii.)$] A path $y_1,\ldots,y_m$ of existential variables, where: both $y_1$ and $y_m$ have edges to \textbf{two} earlier universal variables, and $y_2,\ldots,y_{m-1}$ have edges each to \textbf{one} earlier universal variable.
\item[$(iv.)$] A path $y_1,\ldots,y_m$ of existential variables, where $y_1$ comes before $y_m$, and $y_1$ has an edge to an earlier universal variable. $y_m$ has edges to two earlier universal variables at least one of which comes after $y_1$. Finally, $y_2,\ldots,y_{m-1}$ each have edges to an earlier universal variable.  
\end{itemize}
These cases are illustrated in Figure~\ref{fig:not-yet-drawn}. Using the celebrated result of \cite{RheingoldJACM} it can be seen that one may recognise in logspace whether or not $\Phi$ contains any of these forbidden subinstances. It is not hard to see that if $\Phi$ contains such a subinstance then $\Phi$ is false on $\mathcal{C}_{0111}$ (the universal variables adjacent and before $y_m$ can be played as either $1,2$ or $1,4$ to force $y_m$ to be either $2$ or $4$). We now claim all other instances $\Phi$ are true on $\mathcal{C}_{0111}$ and we demonstrate this by giving a winning strategy for Prover on such an instance. Owing to Case $(i)$ being omitted, Adversary has no trivial win. Prover will now \emph{always play $3$ if she can}. Owing to the omission of Case $(ii)$, Prover never has to answer a variable adjacent to more than two elements. It can be seen that there are few circumstances in which she can not play $3$. Indeed, the only one is if she is forced at some point to play $2$ or $4$ as a neighbour to Adversary's having played $1$. In this case, Adversary can force this response to be propagated as in the chain of cases $(iii)$ and $(iv)$, but because these cases are forbidden, Adversary will never succeed here in winning the game.
\end{proof}
\begin{figure}
\begin{center}
\resizebox{!}{4cm}{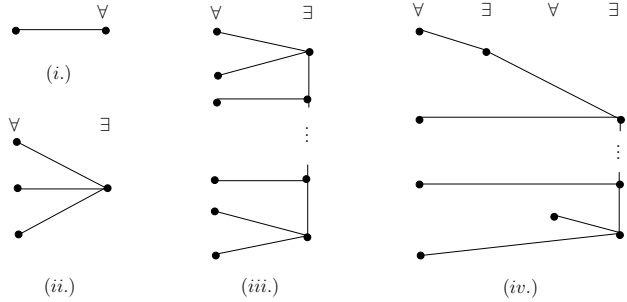}
\end{center}
\caption{Cases from proof of Proposition~\ref{prop:C0111}.}
\label{fig:not-yet-drawn}
\end{figure}

\section{The reflexive cycles}
\label{sec:reflexive}

\begin{figure}
\begin{center}
\resizebox{!}{2cm}{\includegraphics{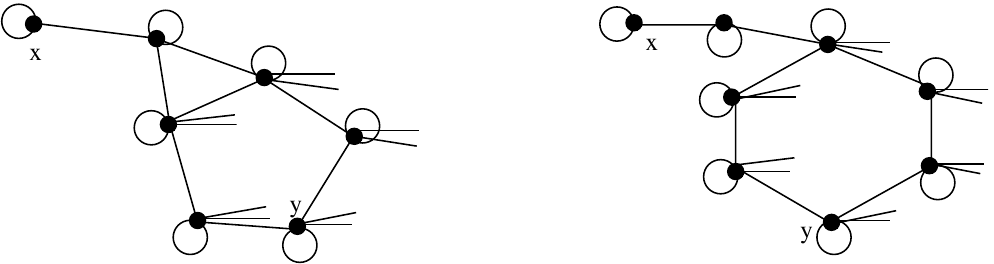}}
\end{center}
\caption{Left ends of the edge gadgets in reductions from QCSP$(\mathcal{K}_5)$ to QCSP$(\mathcal{C}_{1^5})$ and QCSP$(\mathcal{K}_6)$ to QCSP$(\mathcal{C}_{1^6})$. In the former case, the full gadget contains a chain of five copies of $\mathcal{C}_{1^5}$; in the latter case it is a chain of six copies of $\mathcal{C}_{1^6}$}
\label{fig:surhomcref56}
\end{figure}
We will use similar edge gadgets to those of Figure~\ref{fig:surhomcref4} to prove NP-hardness of QCSP$(\mathcal{C}_{1^m})$, for $m\geq 4$. If $m\geq 4$ is even, then the edge gadget $\mathcal{E}_m$ consists of $m$ copies of $\mathcal{C}_{1^m}$ where each copy -- with vertices $1,\ldots,m$, is connected to its successor by edges joining vertex $k$ with vertices $k$ and $k+1$ ($\bmod\ m$). In the first of the copies, the vertex $\frac{m}{2}+1$ is labelled $y$ and a reflexive path of length $\frac{m}{2}-1$ is added to the vertex labelled $1$, which culminates in a vertex labelled $x$. The last of the copies of $\mathcal{C}_{1^m}$ retains the vertex labelling $1,\ldots,m$ -- we consider the other vertices (except for $x$ and $y$) to become unlabelled. Of course,  $\mathcal{E}_4$ is already drawn in Figure~\ref{fig:surhomcref4}. The left end of $\mathcal{E}_6$ is drawn in Figure~\ref{fig:surhomcref56} (to the right). If $m\geq 4$ is odd, then the edge gadget $\mathcal{E}_m$ is drawn in a similar manner, except vertex $\frac{m+1}{2}+1$ becomes $y$ and the reflexive path of length $\frac{m-3}{2}$ that culminates in $x$ at one end and at the other end a vertex that makes a triangle with vertices $1$ and $2$, respectively. The left end of $\mathcal{E}_5$ is drawn in Figure~\ref{fig:surhomcref56} (to the left). These gadgets are borrowed from \cite{FederHell98} and have the property that when the right-hand cycle $\mathcal{C}_{1^m}$ is evaluated automorphically to itself then the rest of the cycles are also evaluated automorphically (but may twist $\frac{1}{m}$th each turn -- this is why we have $m$ copies; $m$ is a minimum number, more copies would still work). Finally, in the left-hand cycle it can be seen that $x$ can be evaluated anywhere except $y$.

Just as in Proposition~\ref{prop:C1111}, we want to try to force vertex-variables $v_1,\ldots,v_m$, corresponding to $1,\ldots,m$, to be evaluated (up to isomorphism) around $\mathcal{C}_{1^m}$.
\begin{proposition}
\label{prop:reflexive}
QCSP$(\mathcal{C}_{1^m})$, for any $m \geq 4$ is NP-hard.
\end{proposition}

\section{Cycles whose loops induce a path}
\label{sec:path}

We begin by recalling the result for QCSP$(\mathcal{P}_{101})$ from  \cite{QCSPforests}, on which our proof for Propositions~\ref{prop:missing} and Corollary~\ref{cor:disconnected} will be based. 
In this proof we introduce the notions of \emph{pattern} and \emph{$\forall$-selector} that will recur in the sequel.
\begin{figure}
\begin{center}
\resizebox{!}{2.5cm}{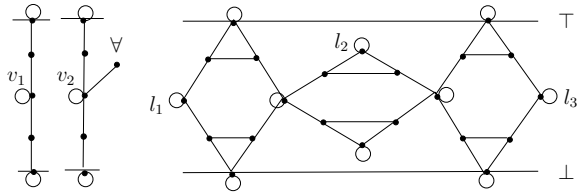}
\end{center}
\caption{Variable and clause gadgets in reduction to QCSP$(\mathcal{P}_{101})$.}
\label{fig:gadgets}
\end{figure}
\begin{proposition}
\label{prop:P101}
QCSP$(\mathcal{P}_{101})$ is Pspace-complete.
\end{proposition}
\begin{proof}
For hardness, we reduce from \emph{quantified not-all-equal 3-satisfiability}, whose Pspace-completeness is well known \cite{Papa}, where we will ask for the extra condition that no clause has three universal variables (of course, any such instance would be trivially false). From an instance $\Phi$ of QNAESAT we will build an instance $\Psi$ of QCSP$(\mathcal{P}_{101})$ such that $\Phi$ is in QNAE3SAT iff $\Psi$ in QCSP$(\mathcal{P}_{101})$. We will consider the quantifier-free part of $\Psi$, itself a conjunction of atoms, as a graph, and use the language of homomorphisms. 

We begin by describing a graph $\mathcal{G}_\Phi$, whose vertices will give rise to the variables of $\Psi$, and whose edges will give rise to the atoms listed in the quantifier-free part of $\Psi$. Most of these variables will be existentially quantified, but a small handful will be universally quantified. $\mathcal{G}_\Phi$ consists of two reflexive paths, labelled $\top$ and $\bot$ which contain inbetween them gadgets for the clauses and variables of $\Phi$. We begin by assuming that the paths $\top$ and $\bot$ are evaluated to vertices $1$ and $3$ in $P_{101}$, respectively (the two ends of $P_{101}$); later on we will show how we can effectively enforce this. Of course, once one vertex of one of the paths is evaluated to, say, $1$, then that whole path must also be so evaluated -- as the only looped neighbour of $1$ in $\mathcal{P}_{101}$ is $1$. The gadgets are drawn in Figure~\ref{fig:gadgets}. The \emph{pattern} is the path $\mathcal{P}_{101}$, that forms the edges of the diamonds in the clause gadgets as well as the tops and bottoms of the variable gadgets. The diamonds are braced by two horizontal edges, one joining the centres of the top patterns and the other joining the centres of the bottom patterns. 
The \emph{$\forall$-selector} is the path $\mathcal{P}_{10}$, which travels between the universal variable node $v_2$ and the labelled vertex $\forall$. (The remainder of this proof is deferred to the appendix.)
\end{proof}
\begin{proposition}
\label{prop:missing}
Let $\mathcal{C}_{0^d1^e}$ be a cycle in which $e > d+3$ ($d$ odd) or $e > d+2$ ($d$ even). Then QCSP$(\mathcal{C}_{0^d1^e})$ is Pspace-complete.
\end{proposition}
\begin{proof}
The reduction is similar to that employed for Proposition~\ref{prop:P101}. We use the pattern $\mathcal{P}_{10^{d}1}$ and $\forall$-selector $\mathcal{P}_{0^{\lfloor \frac{e}{2} \rfloor}1}$. The key part to the reduction is how we get $v_\top$ and $v_\bot$ to evaluate suitably. Let $x_1,\ldots,y_1,\ldots,z_1,\ldots$ be variables not appearing in $\Psi'(v_\top,v_\bot)$ (cf. proof of Proposition~\ref{prop:P101}). In the following, interpret $\lfloor \frac{e-d-5}{2} \rfloor$ to be $0$, if $ \frac{e-d-5}{2}<0$. For $d$ odd, set $\Psi:=\forall x_1 \exists x_2,\ldots,x_{\frac{d+1}{2}}$
\[
\begin{array}{l}
\exists x_{\frac{d+3}{2}},\ldots, x_{\frac{d+3}{2}+\lfloor \frac{e-d-5}{2} \rfloor},\exists v_\top \forall y_1 \exists y_2,\ldots,y_{\frac{d+1}{2}} \exists y_{\frac{d+3}{2}},\ldots, y_{\frac{d+3}{2}+\lfloor \frac{e-d-5}{2} \rfloor}, \exists v_\bot \\
\exists z_1,\ldots,z_{d}  \\
\ [x_1 \Rightarrow x_{\frac{d+3}{2}+\lfloor \frac{e-d-5}{2} \rfloor}, v_\top] \wedge [y_1 \Rightarrow y_{\frac{d+3}{2}+\lfloor \frac{e-d-5}{2} \rfloor}, v_\bot] \wedge [v_\top,z_1,\ldots,z_d,v_\bot] \wedge \\
\mathrm{Ref}(x_{\frac{d+3}{2}},\ldots, x_{\frac{d+3}{2}+\lfloor \frac{e-d-5}{2} \rfloor}) \wedge \mathrm{Ref}(y_{\frac{d+3}{2}},\ldots, y_{\frac{d+3}{2}+\lfloor \frac{e-d-5}{2} \rfloor}) \wedge \Psi'(v_\top,v_\bot)
\end{array}
\]
For $d$ even, set $\Psi:=\forall x_1 \exists x_2,\ldots,x_{\frac{d}{2}}$
\[
\begin{array}{l}
\exists x_{\frac{d+2}{2}},\ldots, x_{\frac{d}{2}+\lfloor \frac{e-d-5}{2} \rfloor},\exists v_\top \forall y_1 \exists y_2,\ldots,y_{\frac{d}{2}} \exists y_{\frac{d+2}{2}},\ldots, y_{\frac{d}{2}+\lfloor \frac{e-d-5}{2} \rfloor}, \exists v_\bot \\
\exists z_1,\ldots,z_{d}  \\
\ [x_1 \Rightarrow x_{\frac{d}{2}+\lfloor \frac{e-d-5}{2} \rfloor}, v_\top] \wedge [y_1 \Rightarrow y_{\frac{d}{2}+\lfloor \frac{e-d-5}{2} \rfloor}, v_\bot] \wedge [v_\top,z_1,\ldots,z_d,v_\bot] \wedge \\
\mathrm{Ref}(x_{\frac{d+2}{2}},\ldots, x_{\frac{d}{2}+\lfloor \frac{e-d-5}{2} \rfloor}) \wedge \mathrm{Ref}(y_{\frac{d+2}{2}},\ldots, y_{\frac{d}{2}+\lfloor \frac{e-d-5}{2} \rfloor}) \wedge \Psi'(v_\top,v_\bot)
\end{array}
\]
\end{proof}
\noindent Note how the previous proof breaks down in boundary cases, for example on the cycle $\mathcal{C}_{0^21^3}$.

The following proofs make use of reductions from Ret$(\mathcal{C})$, where $|C|=m$. It is ultimately intended that the variables $v_1,\ldots,v_m$ in the created instance map automorphically to the cycle. The cycle will be found when the universally quantified $v_1$ is mapped to a non-looped vertex at maximal distance from the looped vertices (sometimes this is unique, other times there are two). We then require that the universally quantified $x_1$ be mapped to a neighbour of $v_1$ at maximal distance from the loops (given $v_1$'s evaluation, this will either be unique or there will be two). All other maps of $v_1$ and $x_1$ lead to degenerate cases. 
\begin{proposition}
\label{prop:odds}
Let $\mathcal{C}$ be an odd $m$-cycle which contains an induced $\mathcal{P}_{11100}$ (or is $\mathcal{C}_{0^21^3}$). Then QCSP$(\mathcal{C})$ is NP-hard.
\end{proposition}
\begin{proposition}
\label{prop:evens}
Let $\mathcal{C}$ be an even $m$-cycle which contains an induced $\mathcal{P}_{11100}$. Then QCSP$(\mathcal{C})$ is NP-hard.
\end{proposition}
\noindent We note that the previous two propositions do not quite use the same techniques as one another. All cases of Proposition~\ref{prop:missing} involving more than one non-loop are weakly subsumed by Propositions~\ref{prop:odds} and \ref{prop:evens} in the sense that Pspace-completeness only becomes NP-hardness.

It is interesting to note that the Proposition~\ref{prop:evens} breaks down on even cycles with two consecutive loops only. It is no longer possible to ensure to encircle the cycle. For these cases we will need yet another specialised construction.
\begin{proposition}
\label{prop:evens2}
For $m \geq 6$, let $\mathcal{C}$ be an even $m$-cycle which contains only two consecutive loops. Then QCSP$(\mathcal{C})$ is NP-hard.
\end{proposition}

\section{Cycles in which the loops induce a disconnected graph}
\label{sec:disconnected}

Let $D_\mathcal{C}:=\{\{\lceil \frac{d_1}{2} \rceil ,\ldots,\lceil \frac{d_m}{2} \rceil\}\}$ be the multiset (of two or more elements), where $d_1,\ldots,d_m$ are the maximal non-looped induced sections (paths) of a cycle $\mathcal{C}$ in which the loops induce a disconnected graph. \mbox{E.g.} a single non-loop between two loops contributes a value $\lceil 1/2 \rceil = 1$ to $D_\mathcal{C}$. We need to split into three cases. 
\begin{proposition}
\label{prop:discon-unique-max}
Let $\mathcal{C}$ be a partially reflexive $m$-cycle in which the loops induce a disconnected graph. If $D_\mathcal{C}$ contains a unique maximal element $\lceil \frac{d}{2} \rceil$, then QCSP$(\mathcal{C})$ is Pspace-complete.
\end{proposition}
\begin{proposition}
\label{prop:discon-even-gaps}
Let $\mathcal{C}$ be a partially reflexive $m$-cycle in which the loops induce a disconnected graph. If $D_\mathcal{C}$ contains only one value, then QCSP$(\mathcal{C})$ is Pspace-complete.
\end{proposition}
\begin{corollary}
\label{cor:disconnected}
Let $\mathcal{C}$ be a partially reflexive $m$-cycle in which the loops induce a disconnected graph. Then QCSP$(\mathcal{C})$ is Pspace-complete.
\end{corollary}

\section{Classification}
\label{sec:class}

\begin{theorem}
\label{thm:over4}
Let $m=d+e\geq 5$. Then QCSP$(\mathcal{C}_{0^d1^e})$ is in NL if I.) $m$ is odd and $e=1$ or $2$, or II.) $m$ is even and $e=0$ or $1$. Otherwise, QCSP$(\mathcal{C}_{0^d1^e})$ is NP-hard.
\end{theorem}
\begin{proof}
Pspace-hardness for irreflexive odd cycles is well-known (see \cite{CiE2006}). Hardness for cycles with disconnected loops follows from Corollary~\ref{cor:disconnected}. Otherwise, for most cycles hardness follows from Propositions~\ref{prop:odds} and \ref{prop:evens}.
For reflexive cycles see Proposition~\ref{prop:reflexive} and for cycles with a single non-loop see Proposition~\ref{prop:missing}. Finally, the outstanding cases of even cycles with two loops are taken care of in Proposition~\ref{prop:evens2}.

Now we address the NL cases. For even cycles with no loops, we are equivalent to QCSP$(\mathcal{K}_2)$ (in NL -- see \cite{CiE2006}). For even cycles $\mathcal{C}_{0^{2i+1}1}$, QCSP$(\mathcal{C}_{0^{2i+1}1})$ is equivalent to QCSP$(\mathcal{P}_{0^{i+1}1})$ (in NL -- see \cite{QCSPforests}). This is because there are surjective homomorphisms from both $(\mathcal{P}_{0^{i+1}1})^2$ to $\mathcal{C}_{0^{2i+1}1}$ and $\mathcal{C}_{0^{2i+1}1}$ to $\mathcal{P}_{0^{i+1}1}$ (see \cite{LICS2008}). For odd cycles, there are surjective homomorphisms from $(\mathcal{P}_{0^{i}1})^2$ to $\mathcal{C}_{0^{2i}1}$ and from $(\mathcal{P}_{0^{i}1})^2$ to $\mathcal{C}_{0^{2i-1}11}$. Thus, QCSP$(\mathcal{C}_{0^{2i}1})$ is equivalent to both QCSP$(\mathcal{C}_{0^{2i-1}11})$ and QCSP$(\mathcal{P}_{0^{i}1})$ and the result follows from \cite{QCSPforests}.
\end{proof}
\begin{theorem}
For $\mathcal{C}$ a partially reflexive cycle, either QCSP$(\mathcal{C})$ is in NL or it is NP-hard.
\end{theorem}
\begin{proof}
Owing to Theorem~\ref{thm:over4}, it remains only to consider partially reflexive cycles on four or fewer vertices. 

Firstly, we consider NL-membership. Each of $\mathcal{C}:=$ $\mathcal{C}_{001}$, $\mathcal{C}_{011}$, $\mathcal{C}_{111}$, $\mathcal{C}_{0000}$, $\mathcal{C}_{0001}$ and $\mathcal{C}_{0011}$ admits a majority polymorphism. It follows form \cite{hubie-sicomp} that QCSP$(\mathcal{C})$ reduces to the verification of a polynomial number of instances of CSP$(\mathcal{C}^c)$, each of which is in NL by \cite{DalmauK08}. Finally, the case $\mathcal{C}_{0111}$ is taken care of in Proposition~\ref{prop:C0111}.

For hardness, it is well-known that \emph{quantified $3$-colouring} QCSP$(\mathcal{K}_3)$ is Pspace-complete \cite{Papa}. And the like result for $\mathcal{C}_{1111}$ appears as Proposition~\ref{prop:C1111}.
\end{proof}
\begin{figure}
\begin{center}
\includegraphics[scale=0.5]{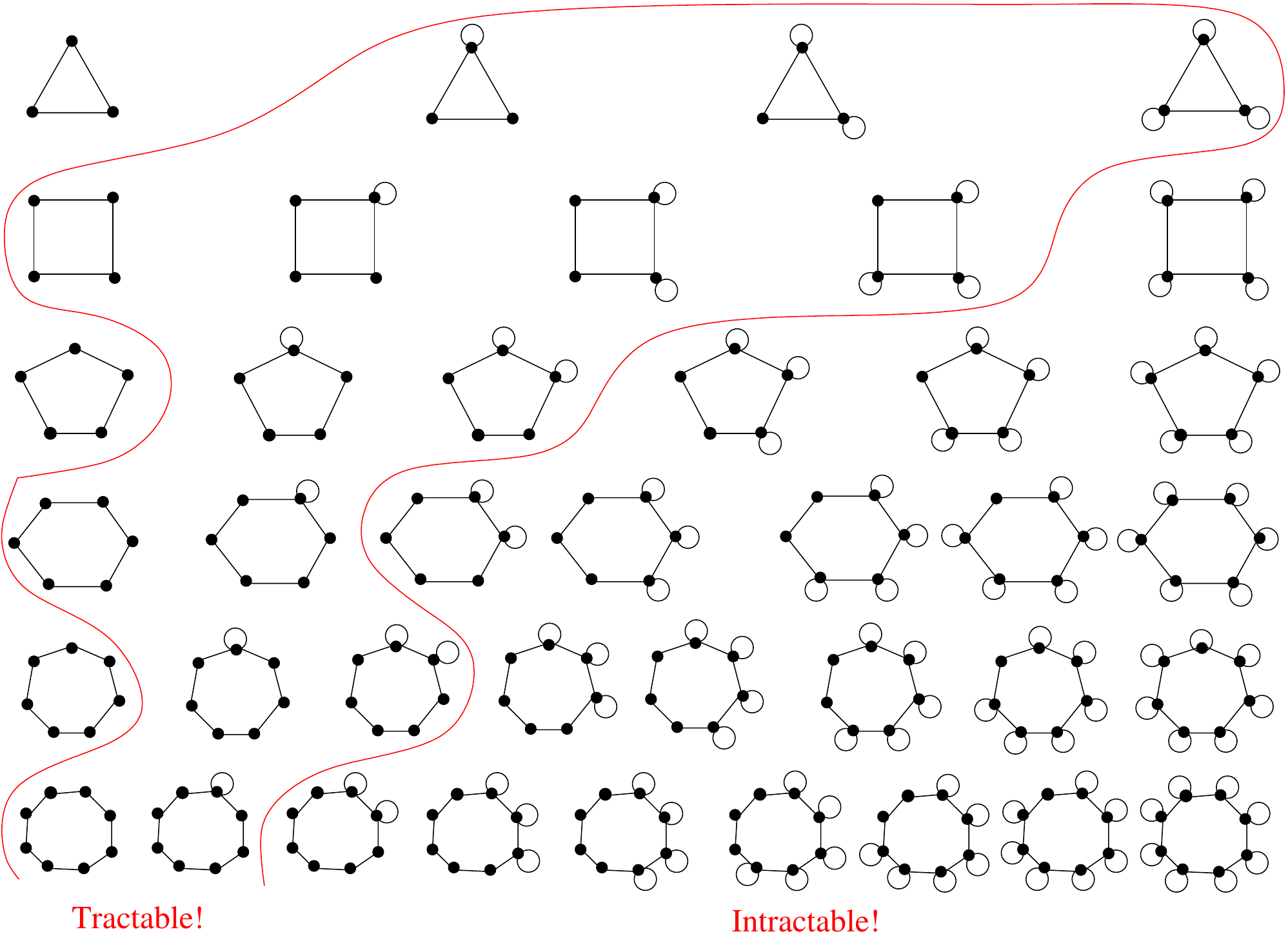}
\end{center}
\caption{The wavy line of tractability.}
\label{fig:wavyline}
\end{figure}

\section{Conclusion}
\label{sec:conclusion}
We have given a systematic classification for the QCSP of partially reflexive cycles. Many of the tractable cases can be explained by the notion of Q-core -- a minimal equivalent QCSP template (see \cite{CP2012}) -- and this is done implicitly in Theorem~\ref{thm:over4}. All NP-hard cases we have seen in this paper have templates that are already Q-cores, with the sole exception of $\mathcal{C}_{0101}$, whose Q-core is $\mathcal{P}_{101}$. By contrast, all of the tractable cases are not Q-cores, except  $\mathcal{C}_{0011}$ and  $\mathcal{C}_{0111}$.

Since finding an algorithm for QCSP$(\mathcal{C}_{0111})$ we became aware of a polymorphism enjoyed by this structure. $\mathcal{C}_{0111}$ has a ternary polymorphism $f$ so that there is $c \in C_{0111}$ such that each of the three binary functions $f(c,u,v)$,  $f(u,c,v)$ and  $f(u,v,c)$ is surjective. The code from Figure~\ref{fig:maroti} will verify such a polymorphism and is intended to run on the excellent program of Mikl\'os Mar\'oti\footnote{See: \texttt{http://www.math.u-szeged.hu{\raise.17ex\hbox{$\scriptstyle\sim$}}maroti/applets/GraphPoly.html}}. The naming of the vertices has been altered according to the bijection $\shefour{3}{0}{1}{4}$ (vertices are numbered from $0$ for the computer).
\begin{figure}
\begin{center}
\texttt{
\begin{tabular}{l}
 arity 3;
 size 4;
 idempotent;
 preserves 01 10 12 21 23 32 30 03 00 11 22; \\
 value 302=0;
 value 320=2;
 value 311=1;
 value 123=1;
 value 223=2;\\
 value 003=0;
 value 032=0;
 value 030=1;
 value 230=2
\end{tabular}
}
\end{center}
\caption{Code for Mar\'oti's program (semicolons indicate new line).}
\label{fig:maroti}
\end{figure}
If follows from \cite{hubie-sicomp} that  $\mathcal{C}_{0111}$ is $2$-collapsible, and hence QCSP$(\mathcal{C}_{0111})$ may be placed in NL by other means.

We would like to improve the lower bound from NP-hardness to Pspace-hardness in the cases of Propositions~\ref{prop:reflexive}, \ref{prop:odds}, \ref{prop:evens} and \ref{prop:evens2}. This might be quite messy in the last three cases, involving careful consideration of the hardness proof for the retraction problem. For the reflexive cycles, though, it just requires more careful analysis of the degenerate cases. This is because we may only have a homomorphic image under $f$ of the cycle for $v_1,\ldots,v_m$, but the universal variables may be evaluated outside of the image of  $f(\{v_1,\ldots,v_m\})$. We will require something of the following form ($\mathcal{E}_m$ is defined at the beginning of Section~\ref{sec:reflexive}).
\begin{conjecture}
\label{lem:gen}
Let $f$ be a function from $\{1,\ldots,m\}$ in $\mathcal{E}_m$ to $\mathcal{C}_{1^m}$ that is a non-surjective homomorphism. Let $f_{ij}$, for some $i,j \in \{1,\ldots,m\}$, be the partial function that extends $f$ from $\{1,\ldots,m,x,y\}$ in $\mathcal{E}_m$ to $\mathcal{C}_{1^m}$, by mapping $x \mapsto i$ and $y \mapsto j$. Then $f_{ij}$ can be extended to a homomorphism from $\mathcal{E}_m$ to $\mathcal{C}_{1^m}$.
\end{conjecture}
\noindent The proof of this seems to be rather technical. For those who doubt it, let us remind ourselves that the length of the chain in $\mathcal{E}_m$ may be any fixed function of $m$, and the reduction of Propositions~\ref{prop:C1111} and \ref{prop:reflexive} will still work. The conjecture is surely easier to prove if we make the chains much longer (say exponential in $m$).

Finally, we conjecture that none of the NL cases are NL-hard, and that most likely our dichotomy can be perfected to L/ Pspace-complete.

\noindent \textbf{Acknowledgements}. The authors thank the referees and are grateful to \mbox{St.} Catherine.

\bibliographystyle{acm}

\pagebreak

\section*{Appendix}

\subsection{Proposition~\ref{prop:reflexive} in full}

\noindent \textbf{Proposition~\ref{prop:reflexive}.}
QCSP$(\mathcal{C}_{1^m})$, for any $m \geq 4$ is NP-hard.
\begin{proof}
We will reduce from the problem CSP$(\mathcal{K}_m)$ in essentially the same manner, for $m>4$, as we did for $m:=4$ and QCSP$(\mathcal{K}_m)$ in Proposition~\ref{prop:C1111}. The salient points remain unchanged, except that we use the gadgets $\mathcal{E}_m$ for edges in the graph input of QCSP$(\mathcal{K}_m)$, instead of $\mathcal{E}_4$, and we may assume that all variables in our reduction instance are existential. So it is that we build a formula $\Psi'(v_1,\ldots,v_m)$ in place of the formula $\Psi'(v_1,\ldots,v_4)$ that we built in Proposition~\ref{prop:C1111}. For $m > 4$, let us define the shorthands (the variables $x'$ and $x_1,\ldots, x_{\frac{m}{2}-2}$ are new and must not appear in $\phi$).
\[
\diamondsuit x \ \phi(x) := \forall x' \exists x_1,\ldots,x_{\lceil \frac{m}{2}\rceil -2} \exists x \ E(x',x_1) \wedge \ldots \wedge E(x_{\lceil \frac{m}{2} \rceil-2},x) \wedge \phi(x)
\]
Recalling that $\mathcal{C}_{1^m}$ is reflexive, it is not hard to verify that $\diamondsuit x$ may be satisfied on $\mathcal{C}_{1^m}$ iff $x$ takes at least $2$ values ($m$ even) or iff $x$ takes at least two non-adjacent values ($m$ odd). To elucidate, let us consider the two examples, when $m:=6,5$.
\[
\begin{array}{ll}
\diamondsuit x \ \phi(x) := \forall x' \exists x_1 \exists x \ E(x',x_1) \wedge E(x_1,x) \wedge \phi(x) & \mbox{if $m:=6$} \\ 
\diamondsuit x \ \phi(x) := \forall x' \exists x \ E(x',x) \wedge \phi(x) & \mbox{if $m:=5$} \\
\end{array}
\]
For $m$ even, we set $\Psi:=$
\[
\begin{array}{ll}
\exists v_1 \forall v_{\frac{m}{2}+1} \exists v_2,\ldots,v_{\frac{m}{2}} \ & [v_1 \Rightarrow v_{\frac{m}{2}+1}] \wedge \diamondsuit v_{\frac{m}{2}+2} \exists v_{\frac{m}{2}+3},\ldots,v_{m} [v_{\frac{m}{2}+1},v_{\frac{m}{2}+2} \Rightarrow v_1] \wedge \\
& \Psi'(v_1,\ldots,v_m)
\end{array}
\]
\noindent For $m$ odd, note that $\exists x \diamondsuit y \ E(x,y)$ selects adjacent pairs (but $x$ and $y$ distinct!). For $m$ odd, we set $\Psi:=$
\[
\begin{array}{ll}
\exists v_1\diamondsuit v_2 \ E(v_1,v_2) \wedge \forall v_{\frac{m+3}{2}} \exists v_3,\ldots,v_{\frac{m+1}{2}},v_{\frac{m+5}{2}}, \ldots, v_m \ & [v_2 \Rightarrow v_{\frac{m+3}{2}}] \wedge [v_{\frac{m+3}{2}} \Rightarrow v_1]  \wedge \\
& \Psi'(v_1,\ldots,v_m) 
\end{array}
\]
It is not hard to see that $\Psi$ forces on some evaluation of $v_1,\ldots,v_m$ that these map isomorphically to $1,\ldots,m$ in $\mathcal{C}_{1^m}$. If they do not, then they instead map according to a homomorphism $f$, and we are in a degenerate case. If they had mapped according to the identity, then a yes-instance (all other variables are existential) extends to homomorphism (from the graph associated with $\Psi$) by, say, $g$. But now $f \circ g$ extends to homomorphism in the degenerate case, and we are done. 
\end{proof}

\subsection{Conclusion of proof of Proposition~\ref{prop:P101}}

For each existential variable $v_1$ in $\Phi$ we add the gadget on the far left, and for each universal variable $v_2$ we add the gadget immediately to its right. There is a single vertex in that gadget that will eventually give rise to a variable in $\Psi$ that is universally quantified, and it is labelled $\forall$ (however, it is the vertex at the other end, labelled $v_2$ in Figure~\ref{fig:gadgets}, that actually corresponds to the universally quantified variable) . For each clause of $\Phi$ we introduce a copy of the clause gadget drawn on the right. We then introduce an edge between a variable $v$ and literal $l_i$ ($i \in \{1,2,3\}$) if $v=l_i$ (note that all literals in QNAE3SAT are positive). We reorder the literals in each clause, if necessary, to ensure that literal $l_2$ of any clause is never a variable in $\Phi$ that is universally quantified. It is not hard to verify that homomorphisms from $\mathcal{G}_\Phi$ to $\mathcal{P}_{101}$ (such that the paths $\top$ and $\bot$ are evaluated to $1$ and $3$, respectively) correspond exactly to satisfying not-all-equal assignments of $\Phi$. The looped vertices must map to either $1$ or $3$ -- $\top$ or $\bot$ -- and the clause gadgets forbid exactly the all-equal assignments. Now we will consider the graph $\mathcal{G}_\Phi$ realised as a formula $\Psi''$, in which we will existentially quantify innermost all of the variables of $\Psi''$ except: 
\begin{itemize}
\item one variable each, $v_\top$ and $v_\bot$, corresponding respectively to some vertex from the paths $\top$ and $\bot$, and
\item all variables corresponding to the centre vertex of an existential variable gadget, and
\item all variables corresponding to the $\forall$-selector (the centre vertex of a universal variable gadget and the vertex labelled $\forall$). 
\end{itemize}
We now build $\Psi'$ by quantifying, adding outermost and in the order of the quantifiers of $\Phi$:
\begin{itemize}
\item existentially, the variable corresponding to the centre vertex of an existential variable gadget,
\item universally, the variable corresponding to the extra vertex labelled $\forall$ of a universal variable gadget, and then existentially, the remaining vertices of the $\forall$-selector.
\end{itemize}
\noindent The reason we do not directly universally quantify the vertex associated with a universal variable is because we want it to be forced to range over only the looped vertices $1$ and $3$ (which it does as its unlooped neighbour $\forall$ is forced to range over all $\{1,2,3\}$). $\Psi'(v_\top,v_\bot)$ is therefore a positive Horn formula with two free variables, $v_\top$ and $v_\bot$, such that, $\Phi$ is QNAE3SAT iff $\mathcal{P}_{101} \models \Psi'(1,3)$. 
Finally, we construct $\Psi$ from $\Psi'$ with the help of two $\forall$-selectors, adding new variables $v'_\top$ and $v'_\bot$, and setting  
\[ \Psi:=\forall v'_\top, v'_\bot \exists v_\top, v_\bot \ E(v'_\top,v_\top) \wedge E(v'_\bot,v_\bot) \wedge \Psi'(v_\top,v_\bot).\] 
The purpose of universally quantifying the new variables $v'_\top$ and $v'_\bot$, instead of directly quantifying $v_\top$ and $v_\bot$, is to force $v'_\top$ and $v'_\bot$ to range over $\{1,3\}$ (recall that $E(v_\top,v_\top)$ and $E(v_\bot,v_\bot)$ are both atoms of $\Psi$). This is the same reason we add the vertex $\forall$ to the universal variable gadget.

We claim that $\mathcal{P}_{101} \models \Psi'(1,3)$ iff $\mathcal{P}_{101} \models \Psi$. It suffices to prove that $\mathcal{P}_{101} \models \Psi'(1,3)$ implies $\mathcal{P}_{101} \models \Psi'(3,1), \Psi'(1,1), \Psi'(3,3)$. The first of these follows by symmetry. The second two are easy to verify, and follow because the second literal in any clause is forbidden to be universally quantified in $\Phi$. If both paths $\top$ and $\bot$ are \mbox{w.l.o.g.} evaluated to $1$, then, even if some $l_1$- or $l_3$-literals are forced to evaluate to $3$, we can still extend this to a homomorphism from $\mathcal{G}_\Phi$ to $\mathcal{P}_{101}$.

\subsection{Latter part of Section~\ref{sec:path} in full}

The following proofs make use of reductions from Ret$(\mathcal{C})$, where $|C|=m$. It is ultimately intended that the variables $v_1,\ldots,v_m$ in the created instance map automorphically to the cycle. The cycle will be found when the universally quantified $v_1$ is mapped to a non-looped vertex at maximal distance from the looped vertices (sometimes this is unique, other times there are two). We then require that the universally quantified $x_1$ be mapped to a neighbour of $v_1$ at maximal distance from the loops (given $v_1$'s evaluation, this will either be unique or there will be two). All other maps of $v_1$ and $x_1$ lead to degenerate cases.

\

\noindent \textbf{Proposition~\ref{prop:odds}.}
Let $\mathcal{C}$ be an odd $m$-cycle which contains an induced $\mathcal{P}_{11100}$ (or is $\mathcal{C}_{0^21^3}$). Then QCSP$(\mathcal{C})$ is NP-hard.
\begin{proof}
We reduce from the problem of Ret$(\mathcal{C})$, known to be NP-complete from \cite{pseudoforests}. Let $\mathcal{C}$ be of the form $\mathcal{C}_{0^d1^e}$. We find a copy of $\mathcal{C}$ with the variables $v_1,\ldots,v_m$ in the following fashion.
\[
\begin{array}{l}
\forall v_1 \exists v_2, \ldots,v_{\frac{m-1}{2}} \\
 \forall x_1 \exists x_2,\ldots,x_{\frac{m+1}{2}} \exists v_{\frac{m+1}{2}}  [v_1 \Rightarrow v_{\frac{m+1}{2}}] \wedge [x_1 \Rightarrow x_{\frac{m+1}{2}}] \wedge \mathrm{Ref}(v_{\lceil \frac{d+2}{2} \rceil},\ldots,v_{\frac{m+1}{2}})  \wedge \\
\mathrm{Ref}(x_{\lceil \frac{d+2}{2} \rceil},\ldots,x_{\frac{m+1}{2}}) \wedge x_{\frac{m+1}{2}}=v_{\frac{m+1}{2}} \wedge \\
\exists v_m,\ldots,v_{\frac{m+3}{2}} \ x_{\frac{m-1}{2}}=v_{\frac{m+3}{2}} \wedge [v_1,v_m \Rightarrow v_{\frac{m+1}{2}}] \wedge  \mathrm{Ref}(v_{\frac{m+3}{2}},\ldots,v_{\lceil \frac{m+2}{2} \rceil}) \wedge \ldots\\
\end{array}
\]
\end{proof}
\noindent It is interesting to see where the previous proof breaks down cycles in which there are less than three loops or a single non-loop. Let us consider the sentence obtained in the case for $\mathcal{C}_{10^4}$.
\[
\begin{array}{ll}
\forall v_1 \exists v_2, v_3 \forall x_1 \exists x_2,x_3 & [v_1 \Rightarrow v_{3}] \wedge [x_1 \Rightarrow x_{3}] \wedge E(v_{3},v_{3}) \wedge x_{3}=v_{3} \wedge \\
& \exists v_5,v_4 \ x_{2}=v_{4} \wedge [v_1,v_5 \Rightarrow v_3] \wedge E(v_{3},v_{3}) \wedge \ldots\\
\end{array}
\]
The problem with $\mathcal{C}_{10^4}$ occurs when $v_1$ and $x_1$ are both evaluated as, say, vertex $3$ (there can be no $v_5$ between $v_1$ and $x_2$).

\

\noindent \textbf{Proposition~\ref{prop:evens}.}
Let $\mathcal{C}$ be an even $m$-cycle which contains an induced $\mathcal{P}_{11100}$. Then QCSP$(\mathcal{C})$ is NP-hard.
\begin{proof}
We reduce from the problem Ret$(\mathcal{C})$, known to be NP-complete from \cite{pseudoforests}. Let $\mathcal{C}$ be of the form $\mathcal{C}_{0^d1^e}$. We find a copy of $\mathcal{C}$ with the variables $v_1,\ldots,v_m$ in the following fashion.
\[
\begin{array}{l}
\forall v_1 \exists v_2, \ldots,v_{\lceil \frac{m-d+2}{2}\rceil} \\
 \forall x_1 \exists x_2,\ldots,x_{\lceil \frac{m-d+2}{2}\rceil}  [v_1 \Rightarrow v_{\lceil \frac{m-d+2}{2}\rceil}] \wedge [x_1 \Rightarrow x_{\lceil \frac{m-d+2}{2}\rceil}] \wedge \mathrm{Ref}( v_{\lceil \frac{m-d+2}{2}\rceil}, x_{\lceil \frac{m-d+2}{2}\rceil}) \wedge \\
 \exists v_{\lceil \frac{m-d+2}{2}\rceil+1}, \ldots,v_{\lceil \frac{m-d+2}{2}\rceil+d-1} \\
 v_{\lceil \frac{m-d+2}{2}\rceil+d-1}=x_{\lceil \frac{m-d+2}{2}\rceil} \wedge  [v_{\lceil \frac{m-d+2}{2}\rceil+1} \Rightarrow v_{\lceil \frac{m-d+2}{2}\rceil+d-1}] \wedge \mathrm{Ref}(v_{\lceil \frac{m-d+2}{2}\rceil+1}, v_{\lceil \frac{m-d+2}{2}\rceil+d-1})  \wedge \\
  \exists v_{\lceil \frac{m-d+2}{2}\rceil+d}, \ldots,v_{m}   [v_{\lceil \frac{m-d+2}{2}\rceil+d-1} \Rightarrow v_{m}]
\end{array}
\]
\end{proof}
\noindent We note that the previous two propositions do not quite use the same techniques as one another. All cases of Proposition~\ref{prop:missing} involving more than one non-loop are weakly subsumed by Propositions~\ref{prop:odds} and \ref{prop:evens} in the sense that Pspace-completeness only becomes NP-hardness.

It is interesting to see where the Proposition~\ref{prop:evens} breaks down on even cycles with two consecutive loops only. It is no longer possible to ensure to encircle the cycle. For these cases we will need yet another specialised construction.

\

\noindent \textbf{Proposition~\ref{prop:evens2}.}
For $m \geq 6$, let $\mathcal{C}$ be an even $m$-cycle which contains only two consecutive loops. Then QCSP$(\mathcal{C})$ is NP-hard.
\begin{proof}
Let us label the vertices of the cycle $1$ to $m$ s.t. the loops are at positions $\frac{m}{2}$ and  $\frac{m}{2}+1$. We introduce the following shorthand ($u'$ does not appear in $\phi$).
\[
\begin{array}{c}
\heartsuit u \ \phi(u) := \forall u' \exists u \ E(u',u) \wedge \phi(u) \\
\end{array}
\]
The $\heartsuit$ quantifier is some kind of weak universal quantifier, with the following important property: if $\phi(2)$ and $\phi(m-1)$ are false then $\phi(1)$ and $\phi(m)$ are true.

Unlike the previous proofs, which used just $v_1,\ldots,v_m$, we will need also $w_1,\ldots,w_n$ and we will be able to say that \emph{at least one of them} maps automorphically to $\mathcal{C}$. We begin with
\[
\begin{array}{l}
\forall v_1 \exists v_2, \ldots,v_{\frac{m}{2}}  [v_1 \Rightarrow v_{\frac{m}{2}}] \wedge E(v_{\frac{m}{2}},v_{\frac{m}{2}}) \wedge \\ 
\forall w_1 \exists w_2, \ldots,w_{\frac{m}{2}}  [w_1 \Rightarrow w_{\frac{m}{2}}] \wedge  E(w_{\frac{m}{2}},w_{\frac{m}{2}}) \wedge
\end{array}
\]
\noindent If $v_1$ and $w_1$ are evaluated on $1$ and $m$, respectively, then $v_{\frac{m}{2}}$ and $w_{\frac{m}{2}}$ will be on $\frac{m}{2}$ and  $\frac{m}{2}+1$, respectively. In the past, we did not care about degenerate cases, so long as they extend to homomorphism, but here we will use the fact that some degenerate cases \emph{could not} extend to homomorphism, if we had used a universal quantifier. Assuming $v_1,w_1$ is mapped to $1,m$, these degenerate cases would come about with $z$ mapped to either $2$ or $m-1$. Thus, using the $\heartsuit$ quantifier we can force $z$ to be mapped to both $1$ and $m-1$. We complete with
\[
\begin{array}{lll}
\heartsuit z & \exists x_2, \ldots,x_{\frac{m}{2}} \exists v_{\frac{m}{2}+1},v_{\frac{m}{2}+2} & E(v_{\frac{m}{2}+1},v_{\frac{m}{2}+1}) \wedge E(v_{\frac{m}{2}+1},v_{\frac{m}{2}+2}) \wedge \\
& & [z,x_2\Rightarrow x_{\frac{m}{2}},v_{\frac{m}{2}+1},v_{\frac{m}{2}+2}] \wedge v_{\frac{m}{2}+1}=x_{\frac{m}{2}} \wedge v_{\frac{m}{2}+2}=x_{\frac{m}{2}-1} \wedge \\
& \exists y_2, \ldots,y_{\frac{m}{2}} \exists w_{\frac{m}{2}+1},w_{\frac{m}{2}+2} & E(w_{\frac{m}{2}+1},w_{\frac{m}{2}+1}) \wedge E(w_{\frac{m}{2}+1},w_{\frac{m}{2}+2}) \wedge \\ 
& & [z,y_2\Rightarrow y_{\frac{m}{2}},w_{\frac{m}{2}+1},w_{\frac{m}{2}+2}] \wedge w_{\frac{m}{2}+1}=y_{\frac{m}{2}}  \wedge w_{\frac{m}{2}+2}=y_{\frac{m}{2}-1} \wedge \\
& & \exists  v_{\frac{m}{2}+3}, \ldots,v_m [v_{\frac{m}{2}+2}\Rightarrow v_m,v_1] \wedge \\
& & \exists  w_{\frac{m}{2}+3}, \ldots,w_m [w_{\frac{m}{2}+2}\Rightarrow w_m,w_1].\\
\end{array}
\]
\noindent The key observation is that the statements $v_{\frac{m}{2}+2}=x_{\frac{m}{2}-1}$ and  $w_{\frac{m}{2}+2}=y_{\frac{m}{2}-1}$ are impossible to satisfy if $v_1,w_1,z$ is mapped to $1,m,2$ or $1,m,m-1$. This is why a universal quantifier can not be used for $z$ and why we can not add $v_{\frac{m}{2}+2}=x_{\frac{m}{2}-1}$ to the construction from the previous proposition.
\end{proof}

\subsection{Section~\ref{sec:disconnected} in full}

Let $D_\mathcal{C}:=\{\{\lceil \frac{d_1}{2} \rceil ,\ldots,\lceil \frac{d_m}{2} \rceil\}\}$ be the multiset (of two or more elements), where $d_1,\ldots,d_m$ are the maximal non-looped induced sections (paths) of a cycle $\mathcal{C}$ in which the loops induce a disconnected graph. \mbox{E.g.} a single non-loop between two loops contributes a value $\lceil 1/2 \rceil = 1$ to $D_\mathcal{C}$. We need to split into three cases. In the following, as in \cite{QCSPforests}, we refer to paths involving loops and non-loops according to words $\{0,1\}^*$ -- where $0$ indicates a non-loop and $1$ a loop. 

\

\noindent \textbf{Proposition~\ref{prop:discon-unique-max}}
Let $\mathcal{C}$ be a partially reflexive $m$-cycle in which the loops induce a disconnected graph. If $D_\mathcal{C}$ contains a unique maximal element $\lceil \frac{d}{2} \rceil$, then QCSP$(\mathcal{C})$ is Pspace-complete.
\begin{proof}
Let $\mathcal{P}$ be the partially reflexive path, obtained by removing the middle (or middle two) vertices in the unique section of non-loops of maximal length in $\mathcal{C}$. It follows from inspection of the proofs from \cite{QCSPforests} that QCSP$(\mathcal{P})$ is Pspace-complete even when the universal variables are restricted to range over some subset of $P$ that contains at least each of the two end-most loops. We reduce from this problem. The $\forall$-selector will be $\mathcal{P}_{0^{\lceil \frac{d}{2} \rceil} 1}$. It remains to show how to relativise the existential variables so they range over precisely $P \subseteq C$. In order to do this, we ensure that each existential variable is at distance $\lceil \frac{d}{2} \rceil -1$ from its nearest loop, by adding a path of length  $\lceil \frac{d}{2} \rceil -1$ to a loop. The reduction is clear and the result follows.
\end{proof}

\

\noindent \textbf{Proposition~\ref{prop:discon-even-gaps}}
Let $\mathcal{C}$ be a partially reflexive $m$-cycle in which the loops induce a disconnected graph. If $D_\mathcal{C}$ contains only one value, then QCSP$(\mathcal{C})$ is Pspace-complete.
\begin{proof}
Let $\mathcal{C}$ be over $n$ vertices, and suppose $D_\mathcal{C}$ contains $k$ instances of its maximal element $f$. These may be produced by sections of non-loops of length $2f$ or $2f-1$. Let $g^+$ be the maximal such length of non-loops. The reduction is similar to that employed in Proposition~\ref{prop:P101}. The pattern will be $\mathcal{P}_{10^{g^+}1}$, while the $\forall$-selector will be $\mathcal{P}_{0^{f} 1}$. It remains to explain how to evaluate $v_\top$ and $v_\bot$. We begin with a universal variable followed by a path of $f$ non-loops and then a loop $(0^f1)$-- just as in the $\forall$-selector. We then have a sequence of alternations of $1^n 0^{2f}$ followed finally by $1^n$. The number of these alternations is $\lceil (\lfloor \frac{k}{2} \rfloor -1)/2 \rceil$ from $v_\top$ and $\lfloor (\lfloor \frac{k}{2} \rfloor -1)/2 \rfloor$ from $v_\bot$. The point is that when the initial universal variables are evaluated at opposite sides of $\mathcal{C}$, with respect to the $k$ large gaps, then some evaluation forces $v_\top$ and $v_\bot$ to be at distance exactly distance $g^+$, across non-loops.
\end{proof}

\

\noindent \textbf{Corollary~\ref{cor:disconnected}}
Let $\mathcal{C}$ be a partially reflexive $m$-cycle in which the loops induce a disconnected graph. Then QCSP$(\mathcal{C})$ is Pspace-complete.
\begin{proof}
Taking into consideration the results of Propositions~\ref{prop:discon-unique-max} and \ref{prop:discon-even-gaps}, the only remaining case is where $D_\mathcal{C}$ has multiple -- $k$ -- copies of its maximal element, and some other smaller elements also. But, this case can be handled very similarly to Proposition~\ref{prop:discon-even-gaps}. Simply amend the alternating motif  $1^n 0^{2f}$ to $1^n (0^{2f-2} 1^n)^n 0^{2f}$ -- thereby allowing any section of $2f-2$ or fewer non-loops to be crossed for free. 
\end{proof}

\end{document}